%% file: paper.tex
\documentclass[sigconf]{acmart}

\usepackage{tikz}
\usepackage{pgfplots}
\usepackage{forest}
\usepackage{xspace}
\usepackage{subcaption}
\usepackage{relsize}
\usepackage{amsmath}
\usepackage{graphicx}
\usepackage[english]{babel}
\usepackage{amsthm}
\usepackage[linesnumbered,ruled,vlined]{algorithm2e}
\newtheorem{theorem}{Theorem}[section]
\usepackage{tikz}
\usepackage{balance}

\usepackage{listings}
\usepackage{float} 

\lstdefinestyle{mystyle}{
    breaklines=true,        
    basicstyle=\ttfamily\footnotesize,  
    frame=single,           
    keepspaces=true,        
    columns=flexible,       
    captionpos=b            
}

\newcommand*\circled[2]{\tikz[baseline=(char.base)]{
            \node[shape=circle,draw,inner sep=2pt, fill=#2] (char) {#1};}}

\AtBeginDocument{%
  }

\setcopyright{acmlicensed}
\copyrightyear{2018}
\acmYear{2018}
\acmDOI{XXXXXXX.XXXXXXX}

\acmConference[Conference acronym 'XX]{Make sure to enter the correct
  conference title from your rights confirmation emai}{June 03--05,
  2018}{Woodstock, NY}
\acmISBN{978-1-4503-XXXX-X/18/06}

\newcommand{\system}{$\lambda$-Tune\xspace}

\begin{document}

\title{\system: Harnessing Large Language Models for Automated Database System Tuning }

\author{Victor Giannakouris}
\affiliation{%
  \institution{Cornell University}
  \city{Ithaca, NY}
  \country{USA}
}
\email{vg292@cornell.edu}

\author{Immanuel Trummer}
\affiliation{%
  \institution{Cornell University}
  \city{Ithaca, NY}
  \country{USA}
}
\email{it224@cornell.edu}

\newcommand\mycommfont[1]{\footnotesize\ttfamily\textcolor{blue}{#1}}
\SetCommentSty{mycommfont}

\SetKwInput{KwInput}{Input}                
\SetKwInput{KwOutput}{Output}              


\begin{abstract}
We introduce \system, a framework that leverages Large Language Models (LLMs) for automated database system tuning. The design of \system is motivated by the capabilities of the latest generation of LLMs. Different from prior work, leveraging LLMs to extract tuning hints for single parameters, \system generates entire configuration scripts, based on a large input document, describing the tuning context. \system generates alternative configurations, using a principled approach to identify the best configuration, out of a small set of candidates. In doing so, it minimizes reconfiguration overheads and ensures that evaluation costs are bounded as a function of the optimal run time. By treating prompt generation as a cost-based optimization problem, \system conveys the most relevant context to the LLM while bounding the number of input tokens and, therefore, monetary fees for LLM invocations. We compare \system to various baselines, using multiple benchmarks and PostgreSQL and MySQL as target systems for tuning, showing that \system is significantly more robust than prior approaches.
\end{abstract}




\newcommand{\rev}[1]{\textcolor{red}{#1}}

\maketitle
\footnotetext{© Victor Giannakouris, Immanuel Trummer | ACM 2025. This is the author's version of the work. It is posted here for your personal use. Not for redistribution. The definitive Version of Record was published in SIGMOD 2025, http://dx.doi.org/10.1145/{number}.}

\input{introduction}
\input{overview}
\input{workload_compression}
\input{configuration_selector}
\input{configuration_evaluator}
\input{experiments}
\input{related}
\input{conclusions}


\bibliographystyle{ACM-Reference-Format}
\bibliography{sample}

\balance

\end{document}

%% file: introduction.tex
\section{Introduction}
The performance of database management system changes dramatically as a function of various tuning choices, including settings for system configuration parameters as well as physical design choices such as indexing, sorting, or partitioning. This has motivated a large body of research on automated database system tuning. Recent work exploits machine learning to find near-optimal configurations~\cite{pavlo2017self, wang2021udo, ding2019ai, giannakouris2022building} but suffers from high training and exploration overheads. This has motivated a new line of research~\cite{trummer2022db, lao2023gptuner}, exploiting LLMs to heuristically prune the search space for tuning. Similar to human database administrators, such models leverage commonsense knowledge, extracted from text documents, to narrow the focus to tuning options that seem ``reasonable'', given the tuning context. This paper presents \system (LAnguage Models for Better Database Administration), a system that exploits capabilities offered by the latest generation of LLMs, including the likes of GPT-4 and Claude~3, to optimize various tuning choices for specific systems and OLAP workloads, including system parameter settings as well as physical design decisions.
\\
\textbf{\system.} Prior approaches to LLM-enhanced database tuning~\cite{lao2023gptuner, trummer2022db} parse text documents (e.g., the database manual) to extract value recommendations for specific parameters. They still need to perform an optimization stage in which hints about specific parameters are combined into complete configurations. This approach is in line with the limitations of early-stage language models such as BERT~\cite{devlin2018bert} and GPT-2~\cite{radford2018improving}. For those models, input and output sizes are limited to a few hundred tokens, restricting the scope of these models to settings for single parameters (rather than entire configurations). Modern LLMs such as GPT-4 support input and output sizes of hundreds of thousands of tokens. The design of \system is motivated by these advances. It exploits increased input sizes by feeding to the language model a description of all information relevant for tuning, including the workload and target system. It also exploits the increased output size by generating entire configurations, rather than hints about single parameters. As shown in our experiments, modern LLMs such as GPT-4 are typically able to map information about the workload to efficient database configuration settings. Hence, unlike prior systems, \system avoids expensive optimization steps, combining settings for single parameters. Instead, it delegates more responsibility to the language model itself.
\\
\textbf{Prompt Generation.} First, \system automates the \emph{prompt generation} step by crafting prompts tailored to the input workload (analytical SQL queries), hardware specifications, and the database system. Our approach incorporates a workload representation method that decomposes the input SQL queries into much smaller, mergeable components, called \emph{query snippets}. As costs increase in the prompt size, minimizing monetary fees while conveying the most relevant information is challenging. We select the most informative subset of snippets to include in the prompt, given a bound on the number of prompt tokens (which are proportional to processing fees for providers like OpenAI). We formulate workload representation as a cost-based optimization problem that we solve by a transformation to integer linear programming. Using the resulting prompt, \system issues multiple calls to the LLM with a certain degree of randomization to obtain multiple candidate configurations. By running the input queries with different configurations, \system evaluates and identifies the most efficient configuration among them using the ideas discussed next.
\\
\textbf{Configuration Selection.} The LLM may return configurations of varying quality. In this context, a challenge is to avoid slowdowns due to particularly bad configurations, incurred, for instance, when evaluating configurations sequentially. To tackle this challenge, we introduce a configuration selection approach that incrementally evaluates the obtained configurations in multiple rounds. Each round comes with a timeout that limits the impact of bad configurations on tuning time. On the other hand, interrupting execution repeatedly may cause redundant work. \system chooses timeouts according to a geometric progression scheme, limiting wasted work due to interruptions before the final round. At the same time, it avoids re-evaluating the same queries across multiple rounds and calculates configuration-specific timeouts, taking into account work accomplished in prior rounds. Reconfiguration overheads, e.g., index creations, may dominate query evaluation time if switching between configurations with a high frequency. Hence, \system adapts query evaluation timeouts to ensure that reconfiguration overheads are proportional to query run time.
\\
\textbf{Configuration Evaluation.} Changing between different configurations can be costly, in particular if it involves index creations. This makes it challenging to keep switching overheads low during evaluations. \system minimizes these overheads by utilizing a \emph{lazy index creation} approach, that only creates the indexes before the execution of a query that might use them, according to the referenced column. At the same time, \system optimally orders query execution according to their index creation costs using a dynamic-programming-based query scheduler, which minimizes query reconfiguration costs when switching between different configurations. Our algorithm is based on a custom cost model we built for our query scheduling needs. We prove that the principle of optimality applies to this cost model in Section~\ref{algo:config_evaluator}.

We evaluate \system over Postgres and MySQL, using the Join Order Benchmark (JOB) and TPC-H as benchmarks. Our experimental evaluation illustrates \system's robustness,  outperforming prior tools for automated database system tuning, including GPTuner~\cite{lao2023gptuner}, DB-Bert~\cite{trummer2023demonstrating} UDO~\cite{wang2021udo}, LlamaTune~\cite{Kanellis2022}, as well as ParamTree~\cite{Yang2023}. In summary, our original scientific contributions are the following:
\begin{itemize}
    \item We present \system, a framework that harnesses Large Language Models for automated, database system tuning for Online Analytical Processing (OLAP) workloads.
    \item We introduce three powerful components that facilitate our LLM-assisted tuning approach, including prompt engineering, configuration selection, and configuration evaluation.
    \item We present a thorough experimental evaluation that showcases that \system is the most robust among its competitors, by consistently identifying the configuration that achieves the best performance.
\end{itemize}

The rest of this paper is organized as follows. Section~\ref{sec:overview} presents an overview of \system, its design, and main goals. Next, in Section~\ref{sec:prompt-gen}, we describe the prompt generation component, which includes our workload compression method. In Section~\ref{sec:config_select}, we present our configuration selection approach. Next, Section~\ref{sec:config_eval} presents the configuration evaluation component. Section~\ref{sec:experiments} presents our experimental evaluation of \system compared to three baselines, as well as an ablation study that showcases the effectiveness of the \system's individual components. Finally, in Sections~\ref{related} and \ref{sec:conclusions} we present the related work before concluding.

%% file: overview.tex
\label{architecture}
\section{\system}
\label{sec:overview}
\textbf{Overview.} \system's architecture is depicted in Figure~\ref{fig:lambda_tune}. \system leverages LLMs to automate the tuning of database systems for OLAP workloads, ensuring to find the configuration that achieves the optimal performance, among the configurations obtained from the LLM. Existing approaches to automated database system tuning depend on the availability of training data, excessive training overheads, and tuning rounds. \system is built on the intuition that the \emph{whole tuning task can be described as a concise prompt} and \emph{downstreamed to an LLM} which already contains and can blend domain-specific tuning information using its pre-trained weights. While not done in the current system, this approach could easily be augmented via retrieval augmented generation, enabling the LLM to parse additional information from the Web. \system takes as input three parameters: an OLAP workload $\mathcal{W} = \{\allowbreak q_1, \allowbreak q_2, \allowbreak ..., \allowbreak q_n\}$ consisting of $n$ queries, a hardware specification $\mathcal{H}$ consisting of the number of cores and memory in the system, and a database system name $\mathcal{D}$. It integrates these three parameters into a prompt tailored for the given setup and obtains configurations from the LLM to optimize the performance of the target system. Optionally, the user can define a token budget $\mathcal{B}$ for the prompt generator, if they wish to restrict the API costs, otherwise, \system will try to fit as much information as possible into the prompt, according to the language model token limit.
\\
\textbf{Tuning Pipeline.} Algorithm~\ref{algo:tuning_pipeline} presents the tuning pipeline of \system. The first step is to pass the input parameters (an OLAP workload, hardware specification, and the target database system) to the prompt generator, described in Section~\ref{sec:prompt-gen}. The prompt generator will first compress the input workload, as described in Section~\ref{sec:prompt-gen:compression}, in order to break down the input queries into smaller text snippets that contain information about specific operators like joins or selections. Then, it will select and combine the most informative snippets for the LLM, with respect to the token budget. Next, the prompt generator transforms and embeds the compressed workload, along with the rest of the input parameters, into a prompt that describes the workload, the database system, and the hardware to the LLM. Next, it invokes the LLM $k$ times, in order to retrieve $k$ responses, each one including a single full configuration. The retrieved configurations will differ according to the degree of randomization of the LLM, determined by the temperature. Each configuration contains a set of SQL commands, compliant with the target database. For instance, if the target system is Postgres, the configuration will typically consist of a list of ``\texttt{CREATE INDEX}'' and ``\texttt{ALTER SYSTEM SET \$param\_name = \$value}'' commands. \system is designed with the assumption that some of the retrieved configurations might be disproportionately slower than the efficient ones. To handle such scenarios, we use an approach that evaluates the retrieved configurations in multiple rounds with a given per-round timeout, preventing inefficient configurations from monopolizing the whole tuning process. This approach provides \emph{provable time guarantees} that are \emph{bounding the tuning time by a function of the optimal execution time (among all configurations retrieved from the LLM)}, as we discuss in Section~\ref{sec:config_select}. Furthermore, to minimize index reconfiguration overheads during evaluation, we first associate indexes with queries that could exploit them, based on column references, and create them lazily, only before an associated query execution. To minimize index reconfiguration costs, we optimally order the queries according to their index generation costs, using a dynamic-programming algorithm presented in Section~\ref{algo:config_evaluator}.

\begin{figure}
    \centering
    \includegraphics[scale=0.35]{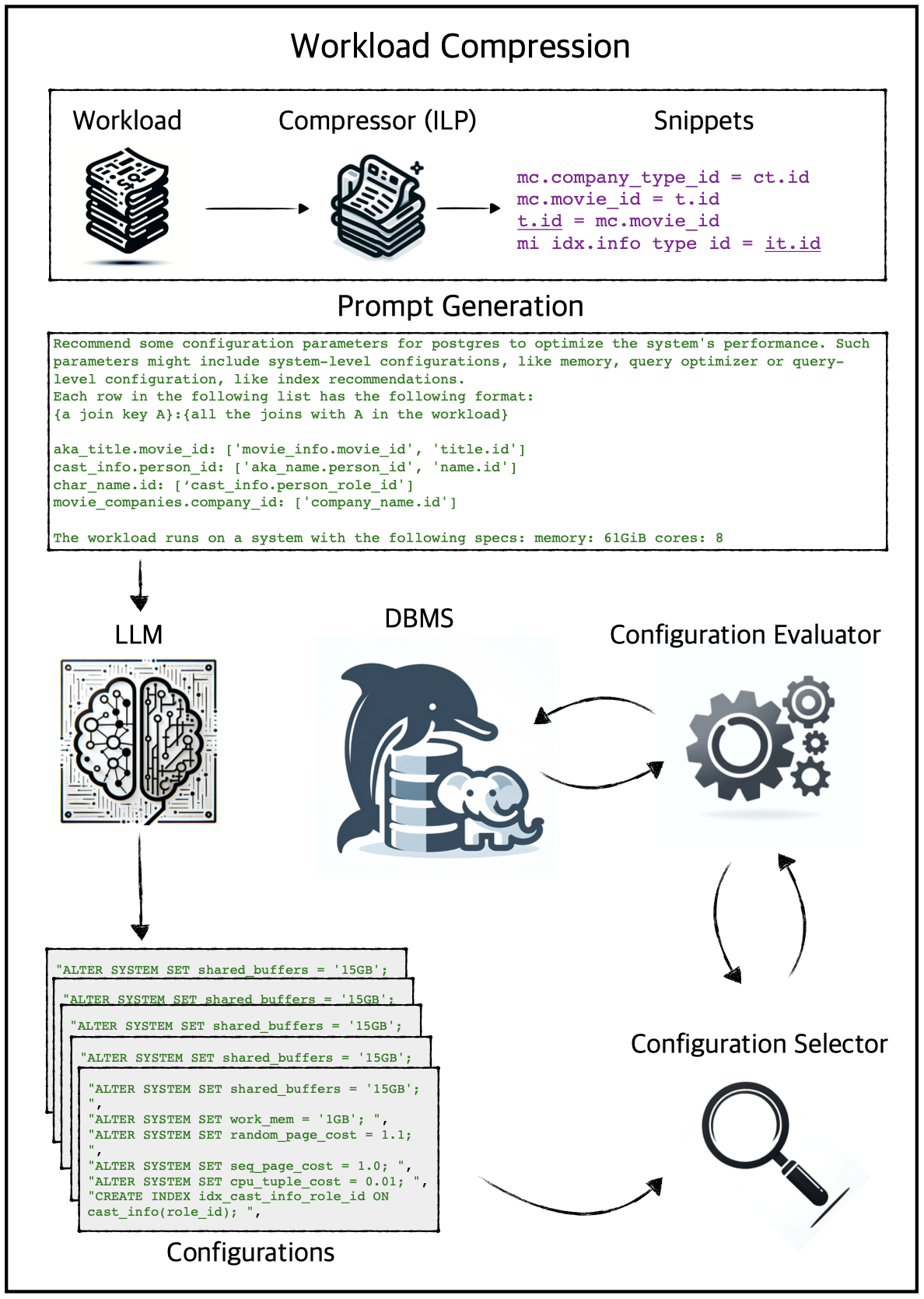}
    \caption{\system Architecture}
    \label{fig:lambda_tune}
\end{figure}

\begin{algorithm}[!ht]
\DontPrintSemicolon
  \SetKwFunction{FConfSel}{Tune}
  \SetKwProg{Fn}{Function}{:}{}
  \Fn{\FConfSel{$\mathcal{W}$, $\mathcal{H}$, $\mathcal{D}$, $\mathcal{B}$}}{
        \tcc{$\mathcal{W}$: The input workload (queries)}
        \tcc{$\mathcal{H}$: The hardware specification (cores, memory)}
        \tcc{$\mathcal{D}$: The database system (e.g. "Postgres", "MySQL")}
        \tcc{$\mathcal{B}$: The token budget}
        \vspace{0.25cm}
        \tcc{Generate the prompt}
        $prompt = GeneratePrompt(\mathcal{W}$, $\mathcal{H}$, $\mathcal{D}$, $\mathcal{B})$\;
        \tcc{Send $n$ API calls to the LLM}
        $C = LLM(prompt, n)$\;
        \tcc{Find the best configuration}
        $best = ConfigSelect(\mathcal{W}, C, t, max\_rounds)$\;
        \vspace{0.25cm}
        \KwRet $best$\;
  }
\caption{\system}
\label{algo:tuning_pipeline}
\end{algorithm}

%% file: workload_compression.tex
\section{Prompt Generation}
\label{sec:prompt-gen}

We describe details of \system's prompt generation process, including the prompt template used as well as our approach for generating a compressed representation of the input workload.

\subsection{Prompt Template}

Listing~\ref{lst:template} depicts our prompt template. Placeholders that are substituted anew for each tuning problem instance are surrounded by curly braces (\verb|${...}|). The prompt template starts with general instructions about the task we expect the LLM to solve: database tuning. The instructions are fairly generic while providing examples of several tuning choices, e.g., related to indexing or memory allocation, that tend to have a significant impact on performance. The placeholder \verb|${DBMS}$| is replaced by the name of the target database management system to tune (e.g., PostgreSQL or MySQL). Due to knowledge gained via pre-training, large LLMs such as GPT-4 are able to adapt the commands for configuration changes to the target system without further instructions.

The next block contains an aggregate description of the input workload. Providing SQL queries directly leads to significant costs for large workloads (since processing fees for LLMs, hosted by providers such as OpenAI, are proportional to input and output sizes). Hence, we provide a compressed representation instead, focusing only on the most important workload aspects while representing information as concisely as possible. As justified in more detail in the following, we focus on describing the join structure of the input workload in the placeholder \verb|${COMPRESSED_WORKLOAD}$|. The preceding sentence provides the LLM with instructions on how to interpret the following, compressed representation.

The final text block of the prompt template contains information on the hardware properties of the target system. For our use cases, we find it sufficient to include the amount of main memory and the number of CPU cores. The prompt template can be extended easily to integrate more details on the hardware.

\begin{lstlisting}[caption={\system Prompt Template}, label={lst:template}, float]
Recommend some configuration parameters for ${DBMS}$ to 
optimize the system's performance. Parameters might 
include system-level configurations, like memory, 
query optimizer or physical design configurations, 
like index recommendations.

Each row in the following list has the following format: 
{a join key A}:{all the joins with A in the workload}
${COMPRESSED_WORKLOAD}$

The workload runs on a system with the following specs: 
memory: ${MEMORY} 
cores: ${CORES}
\end{lstlisting}

\subsection{Workload Compression}
\label{sec:prompt-gen:compression}

We present an approach to summarize input workloads concisely for the LLM, thereby reducing LLM-related processing fees. Our description focuses on binary relationships between different parts of the data. For instance, such relationships could describe the collocation of tables on the same machine in a distributed setting, the co-occurrences of tables in the same queries, or the connections between column pairs that appear in the same join conditions. This type of information is important for tuning decisions such as data partitioning, indexing, and replication. While our approach easily extends to each of the aforementioned properties, we specifically use it to represent join conditions in the current implementation. The motivation for this choice is the fact that joins tend to be among the most expensive operators. Therefore, providing information that helps to optimize the database setup for reduced join overheads is a priority.

Denote by $P\subseteq\{\langle c_1,c_2\rangle|c_1,c_2\in C\}$ the set of pairs of join columns from $C$ that appear together in a join condition in the input workload. The first possible representation is simply the list of these column pairs. However, this representation is sub-optimal as it requires more space than necessary. Instead, it is preferable to merge column pairs that share at least one common column. We use a compressed representation that pairs up a specific column $c_1$ with each column $c_2$ that appears together with $c_1$ in a join predicate. In the prompt, we associate each line of the workload description with one column on the left-hand side, separated by a colon from a comma-separated list of associated columns on the right-hand side. By providing the LLM with instructions on representation semantics (see Listing~\ref{lst:template}), we enable the LLM to correctly interpret the workload summary.

\begin{example}
Assume that $P$ contains the following column pairs: $\langle A,B\rangle$, $\langle A,C\rangle$, and $\langle A,D\rangle$. In a compressed representation, these three binary relationships are summarized in a single line of the prompt: \verb|A:B,C,D|.
\end{example}

Even with the aforementioned compression techniques, it is still not possible to represent the full join structure of large workloads with diverse join conditions. To comply with intrinsic limits on the number of input tokens, associated with all LLMs, as well as with budget constraints of users (since processing more input is more costly), it is necessary to choose which subset of join conditions to represent in the prompt. Our selection strategy is based on the intuition that conditions associated with more expensive joins are more important. The current cost of a join gives an upper bound on how much cost can be reduced via tuning. If the LLM is unaware of the most expensive joins, it cannot effectively reduce processing overheads. Hence, we associate each join column pair $p\in P$ with a value $V(p)$ that represents the total cost associated with joins that use the corresponding condition. We calculate $V(p)$ as the sum $\sum_{j\in J(p)}EC_j$ where $J(p)$ is the set of all join operators in which join condition $p$ is evaluated (considering the default plans chosen by the query optimizer) and $EC_j$ the estimated processing cost, associated with join operation $j$ (this cost can be obtained from the optimizer using corresponding \verb|EXPLAIN| commands). Now, given a limit on the number of tokens used to represent the workload and values $H_c$, representing the number of tokens required to represent column $c\in C$, our goal is to select query snippets, in the form of join conditions, that maximize the accumulated value of join conditions conveyed to the LLM.

\subsection{ILP Formulation}

Picking an optimal combination of join conditions to include in the prompt under a constraint on the number of tokens is a non-trivial problem. Even without considering the possibility of compressing multiple join conditions sharing the same column, it relates to the knapsack problem (weights correspond to token consumption $H$ and utility to the processing cost values $V$) which is NP-hard. Hence, we transform the problem into an integer linear programming problem (ILP) to apply corresponding software solvers.

We introduce binary decision variables $L_c$ for each column $c\in C$, indicating whether or not the corresponding column appears on the left-hand side of a line in the prompt. Also, we introduce binary variables $R_p$ for $p=\langle c_1,c_2\rangle\in P$, indicating whether or not $c_2$ appears on the right-hand side of $c_1$ in a line in the prompt. Clearly, there are dependencies between the two groups of variables. If $R_{\langle c_1,c_2\rangle}$ is set to one, indicating that $c_2$ appears to the right of $c_1$, the associated variable $L_{c_1}$ must be set to one as well: $L_{c_1}\geq R_{p}$ for all $p=\langle c_1,c_2\rangle\in P$. Similarly, if column $c$ appears on the left-hand side, i.e., $L_c=1$, we can prune the search space by imposing at least one associated column on the right-hand side: $L_{c_1} \le \sum_{\langle c_1,c_2\rangle \in P} R_{\langle c_1,c_2\rangle}$. Our goal is to maximize accumulated value while limiting token consumption by budget $\mathcal{B}$:

\begin{align*}
\text{Maximize:} \quad & \sum_{p \in P} V(p) R_p \\
\text{Subject to:} \quad & \sum_{\langle c_1, c_2\rangle \in P}  H_{c_2} \cdot R_{\langle c_1, c_2\rangle} + \sum_{c \in C} H_c \cdot L_c \le \mathcal{B}
\end{align*}

To avoid double-counting pairs of join columns that are symmetric (e.g., \verb|A:B| versus \verb|B:A|), we add one more constraint, avoiding redundant join conditions: $R_{\langle c_1,c_2 \rangle} + R_{\langle c_2,c_1 \rangle} < 2$. Table~\ref{tab:constraints} summarizes all of the aforementioned constraints and variables.

\begin{table}
    \centering
    \caption{Workload Compressor, ILP Constraints}
    \medskip
    \small
    \begin{tabular}{p{0.4\linewidth} p{0.55\linewidth}}
    \toprule[1pt]
    \textbf{Variable} & \textbf{Semantics} \\
    \midrule[1pt]
    $R_{\langle c_1, c_2\rangle} \in \{0, 1\}$ & Binary variable which denotes if $c_2$ appears on the right-hand side of $c_1$.\\
    $L_c \in \{0, 1\}$ & Binary variable which denotes whether the left-hand side column $c$ is included in the prompt.\\
    \toprule[1pt]
    \textbf{Constraint} & \textbf{Semantics} \\
    \midrule[1pt]
    $R_{\langle c_1,c_2\rangle} \le L_{c_1}$ & A right-hand side column $c_2$ can be included only if its left-hand side column is included.\\
    $L_{c_1} \le \sum_{\langle c_1,c_2\rangle \in P} R_{\langle c_1,c_2\rangle}$ & A left-hand side column can be included only if at least one right-hand side column is included.\\
    $R_{\langle c_1,c_2 \rangle} + R_{\langle c_2,c_1 \rangle} < 2$ & Symmetric join snippets cannot be added together.\\
    \bottomrule[1pt]
    \end{tabular}
    \label{tab:constraints}
\end{table}

%% file: configuration_selector.tex
\section{Configuration Selector}
\label{sec:config_select}

The LLM may generate configurations of varying quality. Evaluating those configurations sequentially may lead to large overheads due to particularly slow configurations. Next, we describe how \system avoids this issue.

\textbf{Incremental Timeouts.} Algorithm~\ref{algo:config_selector} describes \system's configuration selection approach. Table~\ref{tab:fields} summarizes the fields of the $ConfigMeta$ objects that are used in several of the following algorithms (during initialization, values for those fields are provided in the order in which they appear in Table~\ref{tab:fields}). To avoid spending too much time evaluating bad configurations, Algorithm~\ref{algo:config_selector} proceeds in rounds and imposes a per-configuration timeout in each round. Initially, it is unclear what timeout allows the best configuration to finish. Hence, Algorithm~\ref{algo:config_selector} increases an initial timeout $t$ gradually by multiplying the timeout with factor $\alpha$ in each round. Having a geometric progression for the timeout is crucial, as it guarantees that the total time spent in all previous rounds (which may be wasted if query evaluation is interrupted due to timeouts) is always proportional to the time spent in the last round (in which at least one configuration finishes executing all queries before the timeout). 

\textbf{Avoiding Redundancy.} Evaluating the same queries with the same configuration is redundant (unless query evaluation is interrupted by timeouts). Hence, \system keeps track of the queries that were fully processed for each configuration. Completed queries are stored for each configuration in the $completedQueries$ field of the $ConfigMeta$ object associated with that configuration. Meta-data about configurations is generally stored in the $configsMeta$ dictionary, mapping configurations to meta-data and initialized in Lines~3 to 5. When selecting queries for evaluation for a specific configuration, \system removes queries already processed (Line~20).

\textbf{Best Configuration.} Algorithm~\ref{algo:config_selector} keeps track of the best currently known configuration in the $best$ variable. This variable is of type $BestConfig$ featuring two fields: the $time$, indicating the execution time of the best currently known configuration, and $config$, describing the best configuration itself. As long as no candidate configurations have been fully evaluated (i.e., all queries have been fully processed with those configurations), the fields of $best$ are set to default values, initialized in Line~2. The while loop from Lines~6 to 13 terminates once at least one configuration has been fully evaluated. It may seem like the first configuration to finish should be the optimal one. That, however, is not generally the case, as illustrated by the following example.

\begin{example}
Consider a workload $W =\{q_1, q_2, q_3\}$, an initial timeout $t = 1$, and two configurations $c_1$ and $c_2$, where the times for queries $q_1$, $q_2$, and $q_3$ are 1, 2, and 4 in $c_1$ and 5, $\epsilon$, $\epsilon$ in $c_2$. In the first three rounds (with timeouts of 1, 2, and 4 seconds respectively), $c_1$ will execute all of the queries, taking 7 seconds in total. However, $c_2$ can achieve a better time of 5 + $2\cdot\epsilon$ seconds, despite the fact that it did not complete any query execution earlier. This means that even if $c_1$ finishes first, configuration $c_2$ is the better configuration overall.
\end{example}

Therefore, \system identifies the optimal configuration as follows. Once the first configuration terminates, Algorithm~\ref{algo:config_selector} uses a different timeout that may, in fact, be higher than the timeout at the start of the corresponding round. \system uses as timeout the execution time of the best currently known configuration (which may be updated repeatedly as other configurations finish) \emph{minus} the time spent fully evaluating queries by the corresponding configuration (Line~19). Any configuration exceeding that timeout is guaranteed to be sub-optimal. The aforementioned timeout can only be set once the total execution time of a first configuration is known. Before that happens, the default timeout of the corresponding round applies. Therefore, once the first configuration finishes, all of the other configurations must be given the chance to finish with the new timeout. This happens in the loop from Lines~14 to 15. After that, the algorithm returns the best configuration. Note that finding good configurations earlier is preferable since it enables tighter timeouts, thereby reducing time wasted on sub-optimal configurations. Therefore, \system iterates over configurations in decreasing order of throughput (i.e., number of queries finished per time unit), assuming that the configuration with the currently highest throughput is more likely to be optimal (which is, of course, not guaranteed).

\begin{example}
    Figure~\ref{fig:config_eval} depicts an example of our configuration selection approach. The x-axis represents execution time, and the y-axis represents the configuration ID. We assume that the timeout increases by factor $\alpha=2$ from one iteration to the next, starting with a timeout of $t=4$. Each configuration is represented with a different color, and a colored square represents a completed query of the configuration of the same row. Gray squares indicate that the last executed query was interrupted due to a timeout. In Round~1, Configuration \circled{1}{pink} completes three queries, Configuration~\circled{2}{yellow} completes one query, and gets interrupted while executing $Q_2$. Configurations \circled{3}{green} and \circled{4}{red} execute two queries and get interrupted while executing $Q_3$. In round two, the timeout doubles to $2\cdot t=8$, meaning that the second round stops after 12 time units total (taking into account the first round as well). Finally, in Round~3, Configuration~\circled{1}{pink} completes all 10 queries after 14 time units total, implying now a configuration-specific timeout for each of the other configurations. None of the remaining configurations terminate within the new timeout. Thus, \system returns Configuration~\circled{1}{pink} as the optimum.
\end{example}
\begin{figure}
        \centering
        \includegraphics[scale=0.4]{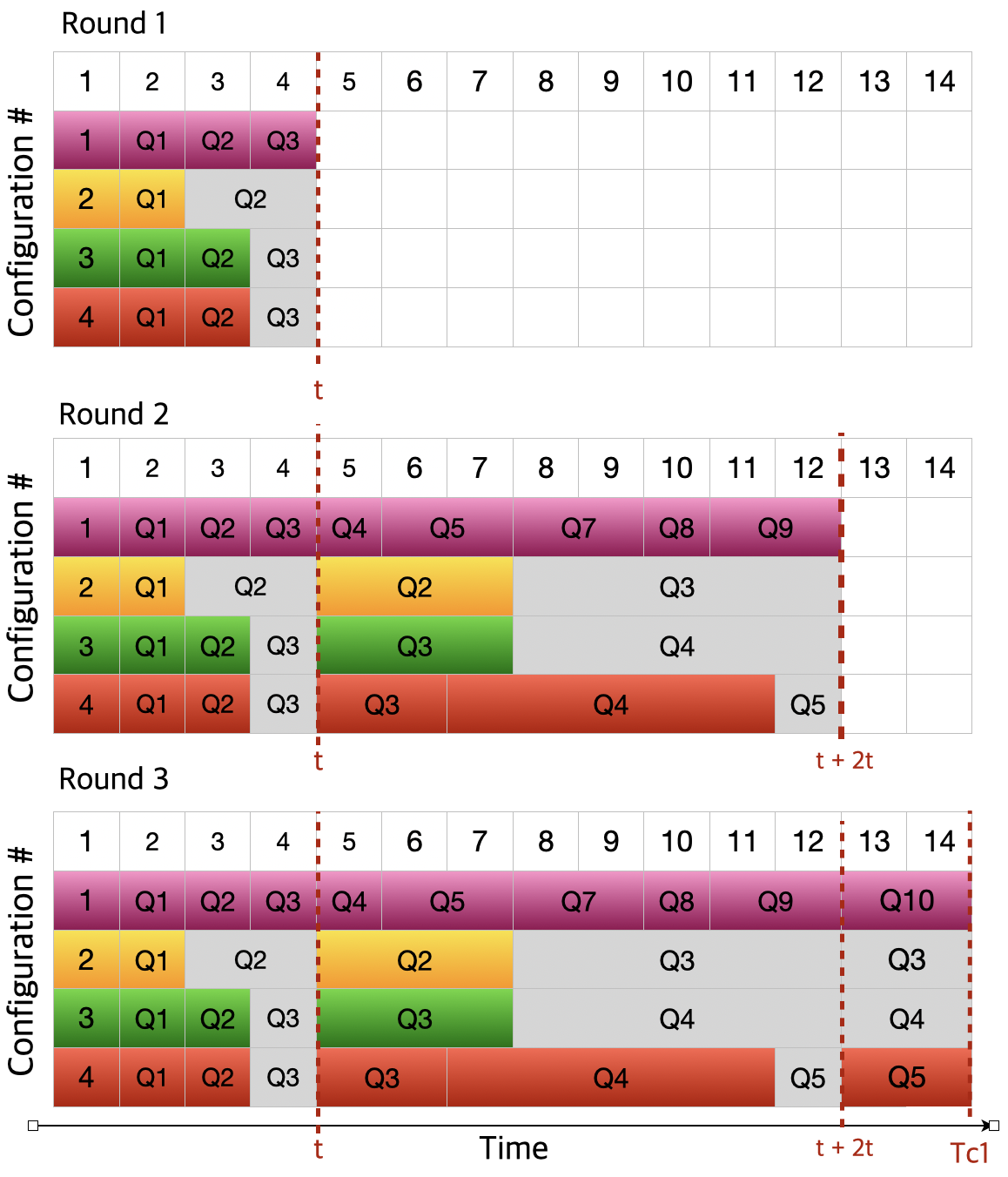}
        \caption{\system Configuration Evaluation }
        \label{fig:config_eval}
\end{figure}

\begin{algorithm}[!ht]
\DontPrintSemicolon
  \SetKwFunction{FConfSel}{ConfigSelect}
  \SetKwProg{Fn}{Function}{:}{}
  \Fn{\FConfSel{$W$, $C$, $t$, $\alpha$}}{
        \tcc{W: The input workload}
        \tcc{C: The input configuration sets}
        \tcc{t: The initial timeout}
        \tcc{$\alpha$: The timeout ratio}
        \vspace{0.25cm}
        \tcc{Initialize best configuration}
        $best \gets BestConfig(\infty,\mathbf{null})$\;
        \tcc{Initialize configuration meta-data}
        $configsMeta \gets dict()$\;        
        \For{$c \in C$} {
            $configsMeta[c]\gets ConfigMeta(0,False,0,\emptyset)$\;
        }
        \tcc{Until first configuration finishes}
        \While{$best.time == \infty$}
        {
            \For{$c \in C$ in decreasing order of throughput}
            {   
                \tcc{Evaluate next queries for this configuration}
                $Update(c, W, configsMeta, t, best)$\;
                \If{$configsMeta[c].isComplete$}
                {
                    $candidates\gets C-\{c\}$\;
                    \textbf{break}\;
                }
                \tcc{Consider re-configuration overheads}
                $t \gets \max_{c\in C}(t, configsMeta[c].indexTime)$ \label{algo:config_selector:adaptive_timeout}
            }
            $t \gets \alpha\cdot t$\; \label{algo:config_selector:alpha}
        }
        \tcc{Check whether remaining configurations are better}
        \For{$c \in candidates$ in decreasing order of throughput}
        {
            $Update(c, W, configsMeta, t, best)$\; \label{algo:config_selector:last_round}
        }
        \KwRet best.config\;
  }
  \vspace{0.25cm}
  \SetKwFunction{FConfSel}{Update}
  \SetKwProg{Fn}{Procedure}{:}{}
  \Fn{\FConfSel{$c$, $W$, $configsMeta$, $t$, $best$}}{
        \tcc{Updates all relevant data structures after query evaluations}
        \If {$best.time \neq \inf$}{
            $t \gets best.time - configsMeta[c].time$\; \label{algo:config_selector:timeout-set}
        }
        \tcc{Execute only the non-executed queries}
        $Q \gets W - configsMeta[c].completedQueries$\;
        $Evaluate(c, Q, t, configsMeta)$
        \vspace{0.25cm}

        \If{configsMeta[c].isComplete}
        {
            \tcc{Keep the best configuration}
            \If{$configsMeta[c].time < best.time$}
            {
                $best.time\gets configsMeta[c].time$\; \label{algo:config_seledctor:best_config}
                $best.config\gets c$\;
            }
        }
  }
\caption{Configuration Selection}
\label{algo:config_selector}
\end{algorithm}

\input{complexity_analysis}

\begin{table}[t]
\caption{Field Names and Descriptions of the $ConfigMeta$ Object Used in Algorithms \ref{algo:config_selector} and \ref{algo:config_evaluator}}.
\begin{tabular}{ll}
\toprule[1pt]
\textbf{Field Name} & \textbf{Description} \\
\midrule[1pt]
$time$ & Completed query time \\
$isComplete$ & Configuration completion flag \\
$indexTime$ & Index creation time \\
$completedQueries$ & Set of completed queries \\
\bottomrule[1pt]
\end{tabular}
\label{tab:fields}
\end{table}
\noindent\textbf{Reconfiguration Overheads.} We have bounded time overheads due to executing queries with different configurations. However, tuning time also depends on the time required for configuration changes. In practice, overheads for index creations tend to dominate reconfiguration overheads. Index creation overheads have the potential to dominate tuning time. For instance, starting tuning with timeouts that are fairly small, compared to index creation overheads, would lead to inefficient tuning, spending more time creating indexes than executing queries. To avoid such cases, \system takes into account index generation overheads when setting timeouts. More precisely, \system measures index generation overheads and adapts timeouts accordingly (see Line~\ref{algo:config_selector:adaptive_timeout} of Algorithm~\ref{algo:config_selector}).

%% file: complexity_analysis.tex
\noindent\textbf{Time Guarantees.} So far, we have justified the choice of timeouts intuitively. Now, we provide a formal proof, showing that our timeout scheme bounds total tuning time for query evaluation as a function of the optimal configuration returned by the LLM.

\begin{theorem}
    \label{theorem:selector}
    The total tuning time (excluding reconfiguration overheads) is in $O(k\cdot \alpha\cdot C_{best})$, where $C_{best}$ is the execution time of the best configuration returned by the LLM, for $\alpha \ge 2$.
\end{theorem}

\begin{proof}
The timeout increases by factor $\alpha$ from one round to the next. Hence, the timeout of the final round, $T_{last}$, is higher than $C_{best}$ at most by that factor: $T_{last}\leq\alpha\cdot C_{best}$. The execution time of each configuration in the last round is upper-bounded by $T_{last}$. Hence, the total execution time of the final round for $k$ configurations is upper-bounded by $k\cdot T_{last}\leq k\cdot \alpha\cdot C_{best}$. However, as the series of timeouts forms a geometric progression with factor $\alpha\geq 2$, the accumulated time of all prior rounds is at most equal to the time of the last round.
\end{proof}

%% file: configuration_evaluator.tex
\section{Configuration Evaluator}
\label{sec:config_eval}
We describe how \system evaluates configurations efficiently while minimizing re-configuration overheads.

\subsection{Evaluating Configurations}

Algorithm~\ref{algo:config_evaluator} shows pseudo-code for the \texttt{Evaluate} procedure, used in Algorithm~\ref{algo:config_selector}. As input, this code takes a configuration $c$ to evaluate, a set $Q$ of queries that have not been fully processed using this configuration, a timeout $t$, and a dictionary mapping configurations (including $c$) to associated meta-data (featuring the fields introduced in Table~\ref{tab:fields} in the previous section).

The effect of executing the procedure is that the meta-data for configuration $c$, stored in $configsMeta[c]$, gets updated with the results of the evaluation. We assume a call-by-reference model such that changes to the aforementioned data structure will be visible to the calling function after executing \texttt{Evaluate}.

Algorithm~\ref{algo:config_evaluator} keeps track of the time remaining for the next queries which is initialized to the input timeout in Line~2 and is updated each time that a query is processed (Line~16). The remaining time is provided as an input parameter to the \texttt{Execute} function, executing single queries (first parameter) with a timeout (second parameter). The \texttt{Evaluate} procedure automatically creates indexes that are required by the current configuration. However, it is inefficient to create all indexes associated with the configuration immediately. Due to the timeout, query evaluation may end long before the last query has finished. This means that creating indexes that are only relevant to later queries is wasteful if those queries are not executed. Note that all indexes created in this procedure are implicitly dropped once the procedure terminates (to enable us to evaluate other configurations without spurious indexes). This means that indexes that are not used would need to be re-created in the next invocation.

Therefore, Algorithm~\ref{algo:config_evaluator} creates indexes \emph{lazily}, creating only the indexes that are potentially relevant for the next query at hand. To determine whether indexes may be relevant, \system analyzes the column and table references in those queries, creating a map $I$ that maps queries to the set of potentially relevant indexes (this map is created in Line~6). To assess whether an index could be useful for a query, we check whether the indexed columns overlap with columns that appear in query predicates. Before evaluating each query, \system creates the indexes possibly needed for this specific query (Line~9). At the same time, \system keeps track of index generation overheads (which are used in Algorithm~\ref{algo:config_selector} to update timeouts). Function~\texttt{createIndexes} returns the time needed to create the corresponding indexes while, at the same time, creating them. Not all relevant indexes must be created as some of them may have been relevant for prior queries, too. Algorithm~\ref{algo:config_evaluator} keeps track of all existing indexes in the $createdIndexes$ variable and subtracts existing indexes from the set of indexes to create (Line~9).

Procedure~\texttt{Evaluate} terminates once the first query is interrupted due to a timeout. In that case, the flag named $complete$, included in the result returned by \texttt{Execute} (this function executes the input query), is set to false. Algorithm~\ref{algo:config_evaluator} checks for interruptions in Line~12. If queries are interrupted, it sets the $isComplete$ flag of the configuration meta-data to false. If no query is interrupted, this field keeps the value of True it is initialized with (Line~4).

Note that Algorithm~\ref{algo:config_evaluator} processes queries in a carefully chosen order, implemented by the \texttt{FindOptimalOrder} function. The following subsections discuss the benefit of that function and its implementation.

\begin{algorithm}[!ht]
\DontPrintSemicolon
  \SetKwFunction{FConfSel}{Evaluate}
  \SetKwProg{Fn}{Procedure}{:}{}
  \Fn{\FConfSel{$c$, $Q$, $t$, $configsMeta$}}{
        \tcc{c: The input configuration}
        \tcc{Q: The input queries}
        \tcc{t: The timeout}
        \tcc{configsMeta: meta-data on configurations}
        $remainingTime \gets t$\;
        $createdIndexes \gets \emptyset$\;
        $configsMeta[c].isComplete\gets True$\;
        $configsMeta[c].indexTime\gets 0$\;
        \vspace{0.25cm}
        \tcc{Creates a query-to-indexes map}
        $I \gets queryIndexMap(Q, c)$\;
        \tcc{Compute the optimal order using Algorithm~\ref{algo:dp-ordering}}
        $Qsorted \gets FindOptimalOrder(Q, I)$\;
        \vspace{0.25cm}
        \tcc{Iterate over queries in optimal order}
        \For{$q \in Qsorted$} {
            \tcc{Create required indexes}
            $configsMeta[c].indexTime += createIndexes(I[q] - createdIndexes)$\; \label{algo:config_evaluator:query-exec:create-indexes}
            $createdIndexes \gets createdIndexes \cup I[q]$\;
            \tcc{Execute next query}
            $queryResult \gets Execute(q, remainingTime)$\; \label{algo:config_evaluator:query-exec}

            \If {\textbf{not} $queryResult.complete$} { \label{algo:config_evaluator:query-exec:timeout}
                $configsMeta[c].isComplete\gets False$\;
                \textbf{break;}\;
            }
            \Else { \label{algo:config_evaluator:query-exec:success}
                $remainingTime -= queryResult.executionTime$\;
                $configsMeta[c].time += queryResult.executionTime$\;
                $configsMeta[c].completedQueries.add(q)$
            }
        }
  }
  \SetKwFunction{FConfSel}{Execute}
        \SetKwProg{Fn}{Function}{:}{}
        \Fn{\FConfSel{$q$, $timeout$}}{
            \tcc{Executes the queries in $q$ with the given timeout.}\;
            \tcc{Returns a $Metrics$ instance which contains the fields $time$ and $completed$, indicating the execution time and the completion of all queries in $q$.}
        }
\caption{Configuration Evaluation}
\label{algo:config_evaluator}
\end{algorithm}

\subsection{Cost Model for Query Scheduling}
\label{sec:config_eval:scheduler}

The order in which queries are processed can have a significant impact on index creation overheads. If all queries were to be executed, index generation overheads would be independent of query order. However, taking into account possible interruptions, more likely to affect later than earlier queries, the order can impact cost. This is illustrated in the following example.

\begin{example}
Assume we process two queries, $q_1$ and $q_2$. Each query can use only one query-specific index whose creation cost is 1 for $q_1$ and 5 for $q_2$. Assume that both queries have the same run time and that the execution of the first query, whichever one it is, gets interrupted with a probability of 50\%. Using the order $q_1-q_2$, the expected cost for index creations is $1+0.5\cdot 5=3.5$, i.e., we definitely must create the index for the first query but incur the cost for the second query index only with a probability of 50\%. Otherwise, if using order $q_2-q_1$, the expected cost is $5+0.5\cdot 1=5.5$. Here, we definitely pay for creating the expensive index of the second query, whereas the index creation overheads of the first query are only paid in 50\% of cases.
\end{example}

As shown in the example, query order matters for performance. \system uses cost-based optimization to order queries. In this subsection, we discuss the cost model it uses. We generalize from the example above. We calculate expected costs for index creations, taking into account different possibilities for the point at which query execution is interrupted due to the timeout. We simplify by assuming that an interruption after each query is equally likely. While highly simplifying, the orders resulting from this model tend to improve performance over default orders, as shown in the experiments. 

Given $n$ queries $q_1$ to $q_n$ to evaluate, we assume that an interruption after each of them (or the case that all of them start execution) is equally likely, i.e., have a probability of $1/n$. Denote by $z_i(Q)$ the index creation overheads for the query $i$, assuming that the set of queries $Q$ has been evaluated before (this is important since the set of indexes to create for the new query depends on which indexes were created for prior queries). Given query order $i_j$ for $1\leq j\leq n$, the cost for interrupting after the $k$-th query is given by $\sum_{1\leq j\leq k}z_{i_j}(\{q_{i_1},\ldots,q_{i_{j-1}}\})$. Hence, the total expected execution cost is given by 
\begin{equation}
1/n\cdot \sum_{1\leq k\leq n}\sum_{1\leq j\leq k}z_{i_j}(\{q_{i_1},\ldots,q_{i_{j-1}}\})\label{eq:expectedCost}
\end{equation}

\subsection{Optimizing Query Order}
To address the problem of optimally ordering query execution and index creation for each configuration evaluation round, we implemented an ordering algorithm based on dynamic programming. Our algorithm is inspired by a classical algorithm for the join ordering problem~\cite{selinger1979access}. Algorithm~\ref{algo:dp-ordering} describes our solution in detail. As input, it takes a set of queries to optimally order, as well as the map $I$ from queries to the set of potentially interesting indexes. The output of Algorithm~\ref{algo:dp-ordering} is the optimal order.

The algorithm is similar to Sellinger's famous dynamic programming algorithm for join ordering~\cite{selinger1979access}. Similarly to the left-deep or right-deep join ordering problem, given $n$ queries, there are $n!$ possible permutations. However, if the first $k$ queries are sorted, the cost of adding the next query (out of the $n-k$ remaining ones) is independent of the order of the first $k$ ones. More formally, we observe that the principle of optimality holds for our cost function. The principle of optimality states that replacing a solution to a sub-problem with a better solution cannot worsen the overall quality. Specifically, in our case, it means that reducing the expected cost of the first $k$ queries by reordering them, cannot worsen the expected cost of all queries.

\begin{theorem}
The principle of optimality holds for expected index creation costs.\label{th:poo}
\end{theorem}
\begin{proof}[Sketch]
In Equation~\ref{eq:expectedCost}, the cost for the first $k$ queries, scaled by a constant, appears as the first $k$ terms in the outer sum. When calculating the expected cost for $k$ queries, we assume that an interruption after each query has a probability of $1/k$. On the other hand, when considering $n>k$ queries, that probability reduces to $1/n$. However, multiplying the first $k$ terms in the cost function for $n$ queries by factor $n/k$ yields exactly the cost of the first $k$ queries alone. As the cost of the first $k$ queries appears, scaled by a positive constant, as a term in the cost function for $n$ queries, the cost for $n$ queries is monotone in the cost of the first $k$ queries. This means that changing the order of the first $k$ queries to reduce the cost for $k$ queries cannot worsen the overall cost. For the remaining terms (after the $k$-th term in the outer sum), the order of the first $k$ queries does not matter anymore.
\end{proof}

This insight motivates Algorithm~\ref{algo:dp-ordering}. We quickly describe it in the following.

Algorithm~\ref{algo:dp-ordering} maintains the optimal cost for query subsets in Variable~$dpCost$. The associated optimal orders are stored in Variable~$dbOrder$. The algorithm initializes both data structures in Lines~4 to 7, using single-element query orders and associating singleton queries with the costs of creating the associated indexes. 

Next, the algorithm enumerates all subsets of queries in increasing order of set cardinality, starting with query pairs and ending with the set containing all queries. For each subset of queries, it evaluates all possibilities to order the queries when using locally optimal orders for query subsets. More precisely, the algorithm considers all possibilities to expand query orders with $k$ elements into query orders with one additional query. Given a subset of queries for which an optimal order should be calculated, it successively considers each query as a candidate to appear last in the corresponding order. Having chosen a query to appear last, Algorithm~\ref{algo:dp-ordering} retrieves the optimal order for the remaining queries (which must have been calculated in prior iterations since Algorithm~\ref{algo:dp-ordering} considers query sets in ascending order of cardinality) and adds the cost of creating the indexes for the last query. Note that the cost of creating indexes for that last query depends on the indexes that have been already generated for the queries that appear first in the order.

Algorithm~\ref{algo:dp-ordering} updates the best order and associated cost for each of the considered options. Whenever the combined cost of creating indexes for prior queries and creating indexes for the last query is below the best currently known cost for a given query subset, the new order and its cost are stored. Finally, the algorithm returns the best order for the entire query set.

\begin{theorem}
Algorithm~\ref{algo:dp-ordering} generates an optimal query order according to our cost model.
\end{theorem}
\begin{proof}[Sketch]
This is a direct consequence of the principle-of-optimality property of our cost function. The algorithm considers all query permutations that use locally optimal solutions for query subsets. However, according to Theorem~\ref{th:poo}, those permutations must contain an optimal solution.
\end{proof}

Note that Algorithm~\ref{algo:dp-ordering} has exponential complexity in the number of input queries as it considers all query subsets. The next section discusses a method by which \system limits the resulting complexity. This ensures that scheduling does not become a significant time factor during the tuning process.

\begin{algorithm}[!ht]
\DontPrintSemicolon
  \SetKwFunction{FConfSel}{ComputeOrderDP}
  \SetKwProg{Fn}{Function}{:}{}
  \Fn{\FConfSel{$W$, $I$}}{
  \tcc{W: The set of queries\\I: A hashmap containing the indexes for each query}
        $dpCost \gets \{\}$\;
        $dpOrder \gets \{\}$\;
        \For{$q \in W$} {
            $indexes \gets I[q]$\;
            $dpOrder[\{q\}] \gets [q]$\;
            $dpTable[\{q\}] \gets cost(indexes)$
        }
        \For{$i=2$ .. $i \le n$}
        {
            \tcc{Enumerate all subsets of size $i$}
            \For{$subset \subseteq W:|subset| = i$} {
                $dpCost[subset] \gets \infty$\;
                \For{$query \in subset$}
                {   
                    $subset' \gets subset - query$\;
                    $queryCost \gets computeCost(query, subset')$\;\label{algo:cs:cost}
                    $c \gets dpTable[subset'] + queryCost$\;
                    \If{$c < dpCost[subset]$} {
                        $dpCost[subset] \gets c$\;
                        $dpOrder[subset] \gets dpOrder[subset'] \circ query$\;
                    }
                }
            }
        }

        \KwRet $dpOrder[set(W)]$\;
  }
\caption{DP Query Scheduling}
\label{algo:dp-ordering}
\end{algorithm}

\subsection{Query Clustering}

Due to the exponential complexity of our query scheduling algorithm, sorting queries for large workloads can be excessively time-consuming. To mitigate this, we reduce the input size to our algorithm by clustering queries according to their index dependencies. To do so, we first need a query vectorization method that complies with standard distance metrics, in order to perform the clustering. Each index is assigned a unique number $i$, and each query $q$ is represented as a binary vector where the $i^{th}$ element indicates whether $q$ references index $i$ (1 if it does, 0 if it does not). We then cluster the queries based on these vector representations, using the Euclidean distance as the clustering distance metric. Given two queries $q$ and $q'$, and their index vectors $I$ and $I'$, then the distance for those two queries can be defined  as $d(I, I') = \sqrt{\sum_{i}^{n}(I_i - I'_i)^{2}}$. Using this metric, we can then proceed with clustering the queries using K-Means. This solution is minimally invasive to our algorithm. For instance, consider two queries, $q_1: A$ and $q_2: A$, both requiring index $A$. As long as the algorithm considers only the index creation cost, these queries can be grouped into a single cluster labeled with index $A$. This approach is also applied to clusters of queries that depend on multiple indexes. We strictly limit the input to our algorithm to a manageable size of 13 queries. 

%% file: experiments.tex
\section{Experimental Evaluation}
\label{sec:experiments}

We compare \system to several baselines on multiple benchmarks and perform an ablation study.

\subsection{Experimental Setup}

All the experiments were executed on an EC2 p3.2xlarge instance, using the Deep Learning Base GPU AMI on Ubuntu 20.04. As benchmarks, we use the TPC-H benchmark with scaling factors one and ten, TPC-DS with scaling factor one, and the Join Order Benchmark (JOB). We tune Postgres 12.0 and MySQL 8.0. As initial configuration, we use the default settings for all system parameters. Unless noted otherwise, no indexes are initially created.

We compare \system to two other tuning systems exploiting LLMs, GPTuner~\cite{lao2023gptuner} and DB-BERT~\cite{trummer2022db}. Also, we compare to UDO~\cite{wang2021udo}, a tuning tool for universal database optimization exploiting reinforcement learning. Furthermore, we compare to LlamaTune~\cite{Kanellis2022}, integrated into the MLOS framework~\cite{Kroth2024}, a system that leverages techniques for dimensionality reduction to improve sample efficiency in automated database tuning. For the problem of index selection, we use two specialized tools as baselines, namely Dexter~\cite{dexter} and the DB2 Index Advisor~\cite{Valentin2000a}. Finally, we compare to ParamTree~\cite{Yang2023}, a system that tunes five parameters used by the PostgreSQL query optimizer: cpu\_tuple\_cost, cpu\_operator\_cost, cpu\_index\_tuple\_cost, seq\_page\_cost, and random\_page\_cost. While ParamTree can optimize settings for those constants on a per-operator level, the PostgreSQL optimizer uses a single value for those parameters for all operators. Hence, for each of the five parameters, we use the average of the operator-specific recommendations. For \system, we set the timeout for the first round to ten seconds and we set $\alpha=10$ as well. For all other baselines allowing to set timeouts (namely UDO and GPTuner), to limit their overheads due to bad configurations, we set this timeout to three times the time of the worst configuration found by \system. UDO executes workload samples to evaluate configurations. Hence, specifically for UDO, we re-execute configurations tried by UDO to measure the execution time for the full workload, thereby making the results comparable to the other baselines. \system uses OpenAI's GPT-4 model to generate configurations.

\subsection{Comparison to Baselines}

\begin{table*}[h!]
\small
\caption{Cost of Best Configuration Found by Each Approach, Scaled to the Cost of the Best Overall Configuration\label{tab:aggregates}}
\begin{tabular}{|cc|c|c|c|c|c|c|c|}
\hline
\multicolumn{1}{|c|}{Benchmark}      & DBMS & Initial Indexes & \system           & UDO                  & DB-Bert       & GPTuner              & LlamaTune                 & ParamTree            \\ \hline
\multicolumn{1}{|c|}{TPC-H 1GB}  & PG & Yes    & 1.07                 & 1.96                 & 1.13          & 1                    & 2.08                 & 3.23                 \\ \hline
\multicolumn{1}{|c|}{TPC-H 1GB}  & MS & Yes    & 1.06                 & 1                    & 1.02          & 1.73                 & 1.39                 & 3.24                 \\ \hline
\multicolumn{1}{|c|}{TPC-H 10GB} & PG & Yes    & 1.03                 & 1                    & 1.05          & 1.04                 & 2.38                 & 3.18                 \\ \hline
\multicolumn{1}{|c|}{TPC-H 10GB} & MS & Yes    & 4.98                 & 1                    & 5.16          & 5.84                 & 2.86                 & 15.2                 \\ \hline
\multicolumn{1}{|c|}{JOB}        & PG & Yes    & 1                    & 1.32                 & 1.05          & 1.1                  & 3.48                 & 3.48                 \\ \hline
\multicolumn{1}{|c|}{JOB}        & MS & Yes    & 1                    & 1.07                 & 3.69          & 3.69                 & 3.22                 & 3.22                 \\ \hline
\multicolumn{1}{|c|}{TPC-H 1GB}  & PG & No     & 1.05                 & 3.76                 & 1             & 1.06                 & 1.43                 & 4.24                 \\ \hline
\multicolumn{1}{|c|}{TPC-H 1GB}  & MS & No     & 1.2                  & 2.83                 & 1.02          & 1                    & 1.61                 & 3.64                 \\ \hline
\multicolumn{1}{|c|}{TPC-H 10GB} & PG & No     & 1.65                 & 1.54                 & 2.45          & 2.52                 & 1                    & 1.54                 \\ \hline
\multicolumn{1}{|c|}{TPC-H 10GB} & MS & No     & 1.04                 & 3.2                  & 1.09          & 1                    & 1.88                 & 3.2                  \\ \hline
\multicolumn{1}{|c|}{JOB}        & PG & No     & 1                    & 1.69                 & 1.08          & 1.13                 & 3.09                 & 3.26                 \\ \hline
\multicolumn{1}{|c|}{JOB}        & MS & No     & 1                    & 3.07                 & 3.07          & 3.07                 & 3.07                 & 3.07                 \\ \hline
\multicolumn{1}{|c|}{TPC-DS}     & PG & No     & 1                    & 1.37                 & 1.67          & 1.66                 & 3.33                 & 3.33                 \\ \hline
\multicolumn{1}{|c|}{TPC-DS}     & MS & No     & 1.79                 & 3.25                 & 1             & 1.03                 & 1.05                 & 3.25                 \\ \hline
\multicolumn{3}{|c|}{\textbf{Average}}                & \textbf{1.41} & \textbf{2.00} & \textbf{1.82} & \textbf{1.91} & \textbf{2.27} & \textbf{4.07} \\ \hline
\end{tabular}
\end{table*}

Figures \ref{fig:exp:original-indexes} and \ref{fig:exp:no-indexes} depict the results of our experimental evaluation. In our plots, the x-axis represents the optimization time (in seconds), while the y-axis represents the best execution time found (in seconds). For instance, the data point $(x, y)$ indicates the best execution time $y$ reported by each framework until time $x$. All experiments were run three times. For each line plot, the middle line represents the average of the best execution time found over different trials, and the shaded area indicates the error range, encompassing the minimum and maximum execution times found by each configuration. Each line starts at the point in time when the corresponding system had evaluated its first configuration. A dashed line is used in cases where the corresponding system did not evaluate any configurations successfully over the tuning time.

Figure~\ref{fig:exp:original-indexes} restricts the tuning scope to system parameter tuning. This means tuning approaches cannot change the physical design by creating indexes. Instead, all tuning methods use the same indexes, created before tuning starts and covering primary key and foreign key columns referred to in the input workload. All system parameters are initially set to the default values. Figure~\ref{fig:exp:no-indexes} expands the tuning scope, allowing baselines to change the physical design as well as settings for system parameters (starting without any indexes and with the default settings for all system parameters). Some of the baselines, namely UDO and \system, tune parameter settings as well as the physical design. The other baselines focus on system parameters tuning alone. For those systems, we create indexes recommended by Dexter~\cite{dexter} before tuning starts (we omit tuning results for those baselines without Dexter's indexes as those results are uniformly worse). Table~\ref{tab:aggregates} summarizes the results in the two aforementioned figures, reporting the scaled cost of the best configuration found by each baseline. For each scenario, we scale costs to the cost of the optimal configuration found for this scenario by any of the baselines. The ``Initial Indexes'' column indicates whether indexes are generated before tuning starts (restricting the tuning scope to parameters) or not.

As shown in Table~\ref{tab:aggregates}, \system is the most robust tuning method on average, followed by DB-BERT and GPTuner. Interestingly, this means that the three baselines exploiting tuning hints gained from text rank first when averaging over all scenarios. However, DB-BERT and GPTuner still explore a combinatorial search space that combines different hints mined from text. \system explores a much smaller space, consisting only of the few complete configurations generated by the language model. As demonstrated by the experimental results, this enables \system to find promising configurations faster. \system differs from both GPTuner and DB-BERT in that it also recommends indexes (in addition to system parameter settings). Averaging relative costs of the configurations found by \system in the scenarios allowing index creations only, its relative cost decreases to 1.21 (versus 1.41 when averaging over all scenarios). This shows that \system's expanded tuning scope provides additional benefits. UDO is another baseline that optimizes for physical design and parameter settings at the same time. However, lacking hints from text to heuristically guide tuning choices, it converges to optimal decisions more slowly. LlamaTune achieves the best result for TPC-H on PostgreSQL with a scaling factor of 10 and near-optimal results in a few other scenarios. However, LlamaTune does not heuristically constrain the search space via mined tuning hints and suffers from configurations with high run times in some scenarios. LlamaTune has been previously evaluated on OLTP workloads only, using a fixed evaluation time. This scenario makes mechanisms that limit evaluation time spent with bad configurations and OLAP workloads (such as the one implemented by \system) unnecessary. ParamTree only explores parameter settings with a limited scope (i.e., for the query optimizer), preventing it from tuning other parameters with significant performance benefits.

Table~\ref{tab:aggregates} only considers the quality of the best configuration found over the entire tuning time, Figures~\ref{fig:exp:no-indexes} and \ref{fig:exp:no-indexes} provide more details, showing that \system tends to find optimal configuration also significantly faster, compared to other baselines. This demonstrates the benefits of evaluating complete configurations only.

 \begin{figure}
     \centering
     \includegraphics[width=0.95\linewidth]{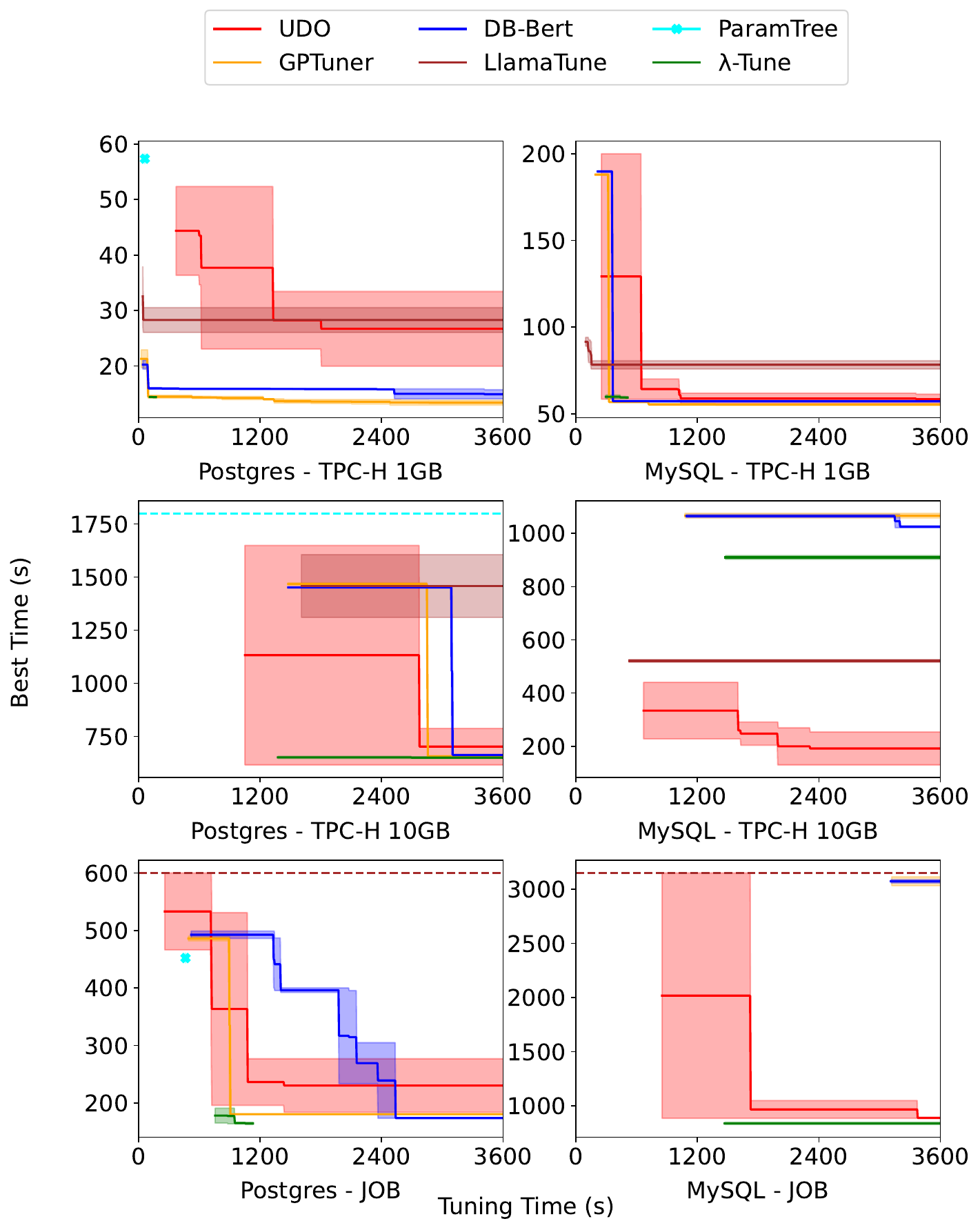}
     \caption{Scenario 1: Baselines do not Create Indexes (Pure Parameter Tuning), Default Indexes Available}
     \label{fig:exp:original-indexes}
 \end{figure}

\begin{figure}
 \centering
 \includegraphics[width=1\linewidth]{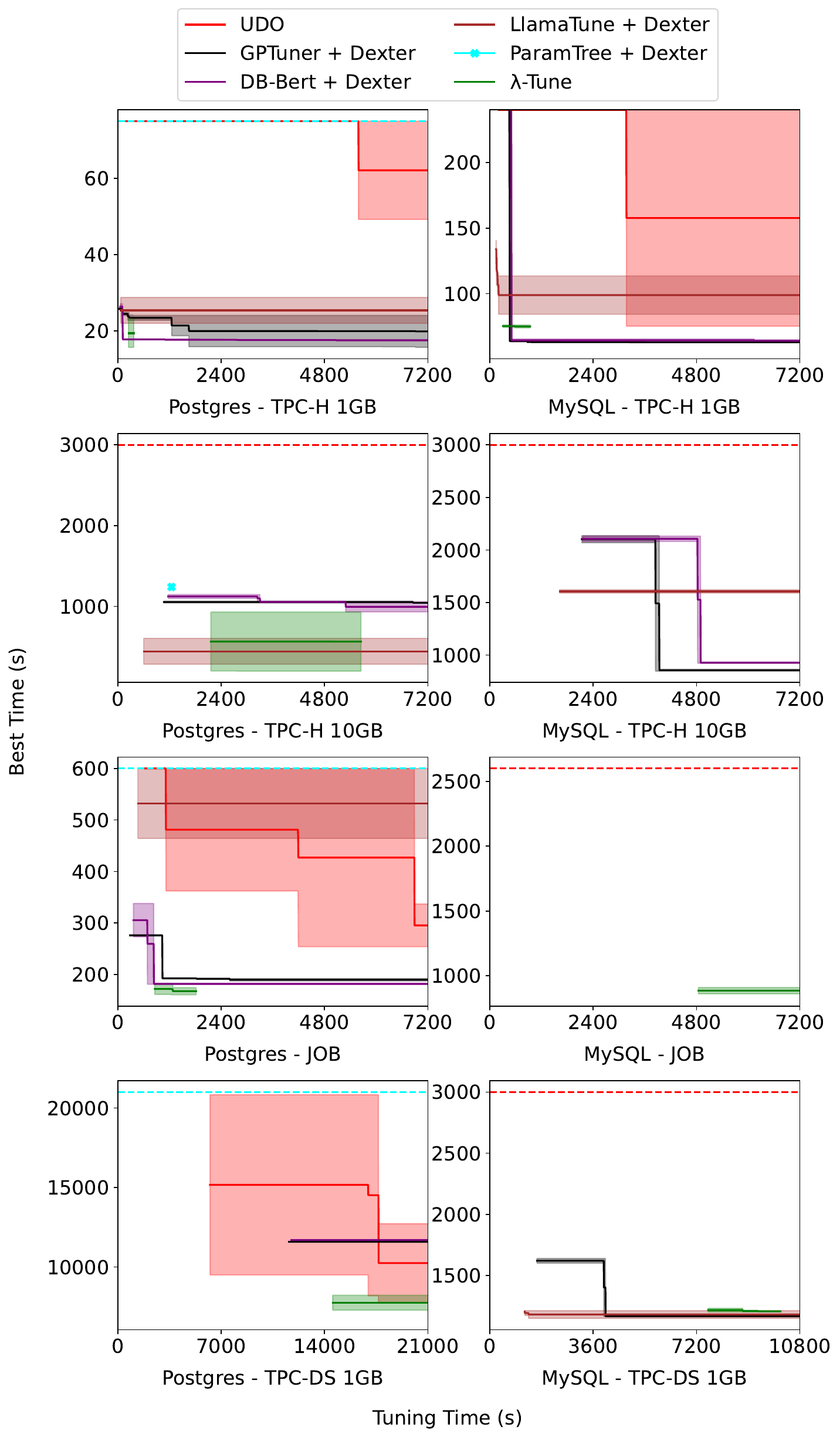}
 \caption{Scenario 2: Baselines Create Indexes, no Indexes are Created by Default}
 \label{fig:exp:no-indexes}
\end{figure}

\begin{table*}[h!]
\small
\caption{Number of Configurations Evaluated per Baseline (Postgres)\label{tab:NrEvaluated}}
\begin{tabular}{|l|l|r|r|r|r|r|r|}
\hline
Scenario   & Initial Indexes & \multicolumn{1}{l|}{\system} & \multicolumn{1}{l|}{UDO} & \multicolumn{1}{l|}{DB-Bert} & \multicolumn{1}{l|}{GPTuner} & \multicolumn{1}{l|}{LlamaTune} & \multicolumn{1}{l|}{ParamTree} \\ \hline
TPC-H 1GB  & Yes             & 5                               & 617                      & 115                          & 103                          & 10                        & 1                              \\ \hline
TPC-H 1GB  & No              & 5                               & 707                      & 171                          & 156                          & 19                        & 1                              \\ \hline
TPC-H 10GB & Yes             & 5                               & 83                       & 6                            & 3                            & 1                         & 1                              \\ \hline
TPC-H 10GB & No              & 5                               & 120                      & 6                            & 5                            & 4                         & 1                              \\ \hline
\end{tabular}
\end{table*}

\begin{table}[h]
    \small
    \centering
    \caption{Best \system Configuration for TPC-H 1GB (Postgres)}
    \begin{tabular}{|l|l|l|}
        \hline
        \textbf{Parameter} & \textbf{Category} & \textbf{Value} \\ \hline
        \verb|shared_buffers| & Memory & 15GB \\ \hline
        \verb|work_mem| & Memory & 1GB \\ \hline
        \verb|effective_cache_size| & Optimizer & 45GB \\ \hline
        \verb|maintenance_work_mem| & Memory & 2GB \\ \hline
        \verb|checkpoint_completion_target| & Logging & 0.9 \\ \hline
        \verb|wal_buffers| & Logging & 16MB \\ \hline
        \verb|default_statistics_target| & Optimizer & 100 \\ \hline
        \verb|random_page_cost| & Optimizer & 1.1 \\ \hline
        \verb|effective_io_concurrency| & IO & 200 \\ \hline
    \end{tabular}
    
    \vspace{0.25cm}
    
    \begin{tabular}{|l|l|}
        \hline    
        \textbf{Table} & \textbf{Indexed Columns} \\ \hline
        lineitem & l\_suppkey, l\_orderkey, l\_partkey \\ \hline
        orders   & o\_custkey, o\_orderkey \\ \hline
        part     & p\_partkey \\ \hline
        partsupp & ps\_partkey, ps\_suppkey \\ \hline
        supplier & s\_nationkey, s\_suppkey \\ \hline
        customer & c\_custkey, c\_nationkey \\ \hline
        nation   & n\_nationkey, n\_regionkey \\ \hline
    \end{tabular}
    \label{tab:postgres_combined_with_indexes}
\end{table}

\begin{figure}
    \centering
    \includegraphics[width=1\linewidth]{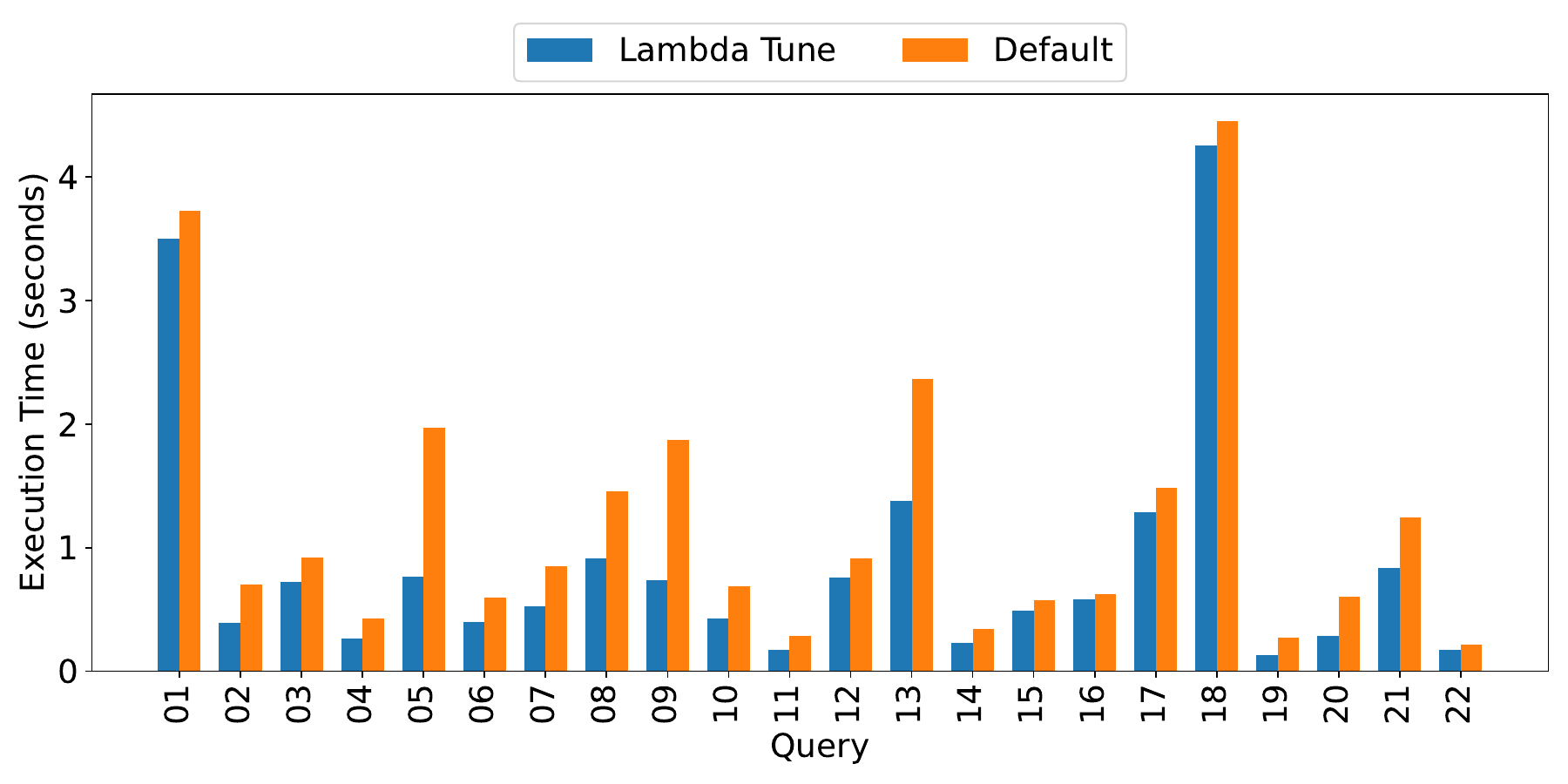}
    \caption{Query Execution Times (TPC-H 1GB, Postgres): \system vs Default Configuration}
    \label{fig:exp:query-times}
\end{figure}

\subsection{In-Depth Analysis}

We analyze the results for one benchmark, TPC-H on Postgres, in more detail. Table~\ref{tab:postgres_combined_with_indexes} shows the configuration selected by \system in detail. The upper part of Table~\ref{tab:postgres_combined_with_indexes} shows changes to default parameter settings, recommended by \system. The parameters are classified according to broad categories (e.g., memory and optimizer). Clearly, compared to the default settings, \system increases the values for several parameters representing the amount of memory reserved for the database management system. This is often beneficial for OLAP workloads. Note that the settings for the \verb|shared_buffers| parameter matches the recommendation from the Postgres manual, stating ``a reasonable starting value for shared\_buffers is 25\% of the memory in your system.''. It seems \system is able to apply this recommendation to the hardware specification of the target system, featuring 61~GB of RAM, as outlined in the hardware description of the input prompt. 

Besides memory-related parameters, \system changes several settings for optimizer-related parameters. In particular, several changes, namely increasing the value of \verb|effective_cache_size| and decreasing the value of \verb|random_page_cost|, compared to the default settings, tend to motivate the query optimizer of Postgres to use indexes more often. This aligns with the fact that \system proposes multiple indexes as part of the same recommended configuration. 

Also, \system increases the value of \verb|effective_io_concurrency|, compared to the default setting. This change has been reported to occasionally lead to significant performance improvements for large scans\footnote{\url{https://stborden.wordpress.com/2022/12/27/tuning-the-postgresql-effective_io_concurrency-parameter/}}. A subset of parameter changes refer to logging behavior which is less relevant for the benchmark, showing that \system might benefit from additional workload-related information in the prompt.

The lower part of Table~\ref{tab:postgres_combined_with_indexes} reports the indexes created by \system, reporting the table and the corresponding column (all indexes created by \system for Postgres were tree indexes with single-column search keys). Clearly, \system focuses on columns and tables that appear frequently in the query workload.

Index recommendations are generally not transferable between workloads on different databases. However, system parameter settings could, in principle, be transferred between different OLAP workloads. We compare the parameter settings proposed by \system across the three benchmarks (TPC-H, TPC-DS, and JOB). It turns out that settings for several memory-related parameters tend to be similar (the optimal configurations use the same settings for \verb|shared_buffers| and \verb|maintenance_work_mem|). Also, whenever \system adds commands for creating indexes, it encourages the optimizer to use them by changing values for \verb|effective_cache_size| and \verb|random_page_cost|. On the other hand, settings for other parameters differ. For instance, only for TPC-DS does \system choose to increase the \verb|max_parallel_workers_per_gather| parameter value beyond the default. The recommendations generated by the LLM for TPC-H tend to overlap in many aspects (e.g., they consistently set \verb|shared_buffers|) but in 15 LLM samples for the TPC-H prompt, we observe outlier configurations where the run time is up to five times higher than the optimum. This underlines the need for configuration selection with bounded execution time.

Table~\ref{tab:NrEvaluated} reports the number of trial runs executed by each baseline. This metric shows significant differences between the different baselines. ParamTree recommends fixed settings for optimizer constants, requiring only a single workload evaluation (after training). \system evaluates only the five configurations proposed by the LLM. For TPC-H with scaling factor one, this makes \system one of the most sample-efficient baselines, followed closely by LlamaTune (which increases sample-efficiency via dimensionality reduction). DB-BERT and GPTuner evaluate a significantly higher number of configurations for scaling factor one. UDO evaluates most configurations. However, UDO performs each evaluation with a workload sample (i.e., it does not run all of the queries), meaning that its measurements are not always representative for the quality of a configuration on the full workload.
When increasing the scaling factor to 10, the number of trial runs decreases significantly for all baselines except for ParamTree and \system. This is expected as each run takes longer, due to the increasing input data size.

Finally, we compare per-query execution times between the default settings and the configuration chosen by \system for TPC-H. Figure~\ref{fig:exp:query-times} reports corresponding results. It turns out that the performance gain via the configuration proposed by \system translate to gains or at least equal performance, compared to the default settings, for each single query.

\input{ablation}

%% file: ablation.tex
\subsection{Ablation Study}
\label{sec:ablation}

To showcase the effectiveness of the multiple individual components of \system, we present the results of an ablation study in which we switched off different components of \system to measure the resulting performance degradation. The experiments presented next focus on tuning Postgres for the JOB benchmark. Figure~\ref{fig:ablation-postgres} depicts the results of our study. The green line (labeled ``Default'') represents the performance of \system with all of its components enabled. Next, we discuss the impact of specific changes to \system.

\subsubsection{Adaptive Timeout}
First, we turn off the component that adaptively sets the timeout of the configuration selection component according to the index creation overheads, discussed in detail in Section~\ref{sec:config_select} (Reconfiguration Overheads). Doing so prevents the system from taking into account reconfiguration overheads when choosing timeouts for query execution. This means that query execution overheads may be dominated by reconfiguration overheads, making the tuning approach inefficient. Indeed, as shown in the plot, this change increases the time needed to find near-optimal configurations from less than 1,000 to over 1,300 seconds. While reconfiguring the system too frequently slows down tuning, it does not decrease the quality of the configurations found by \system.

\subsubsection{Query Scheduler}
Next, we showcase the effectiveness of our query scheduler. In turning this component off, we disable \system's capability to minimize index creation overheads by optimally ordering queries and creating indexes only if they are immediately relevant. Turning off this component increases the time until the first configuration is completely evaluated (i.e., all queries have finished processing) to more than 1,800 seconds. While this change affects the time needed by \system to report first evaluation results for its configurations, it does not degrade the quality of the configurations that \system ultimately returns.

\subsubsection{Obfuscated Workload}

JOB and TPC-H are popular workloads that are likely to appear in the pre-training data of LLMs such as GPT-4. Hence, apriori, it is unclear whether LLMs generate configurations that appear on the Web. If so, obfuscating the workload should significantly impact \system's performance. To test that hypothesis, we hid the names of tables and columns in the input workload. We replaced all table and column names in the extracted query snippets with generic identifiers (e.g., ``Tx'' and ``Cy'' where $x$ and $y$ are integer IDs and ``T'' and ``C'' indicate tables and columns respectively). It is worth to mention that we obfuscate the snippets after the extraction (we do not provide full queries to the LLM), and thus, it could not imply the benchmark from the query templates as this information is lost after the snippet extraction. The results are reported as ``Obfuscated Workload'' in Figure~\ref{fig:ablation-postgres}. However, the results are inconsistent with our hypothesis since the performance remains virtually equivalent to the default settings. This provides evidence that \system does not benefit from hints in the pre-training data.

\subsubsection{Compressor}
Describing the input workload by submitting SQL queries as part of the prompt may seem like the most natural approach. Instead, we opt to use a component that compresses the join structure of the input workload. We evaluate performance when switching that compressor off, adding instead as many SQL queries as we can fit into the prompt with the intrinsic token limits of the LLM. For this specific workload, we are able to fit in 26 full SQL queries. However, doing so increases the time until the first configuration is completely evaluated and also increases the execution time of the best configuration found. Furthermore, the number of input tokens, and therefore processing fees due to LLM invocations, increase (this is analyzed in more detail in the plot discussed next).

\subsubsection{Token Budget}
Figure~\ref{fig:ablation-compressor} analyzes the impact of the prompt structure and token limits in more detail. All approaches are labeled with the number of tokens consumed for workload representation (between parentheses). Clearly, adding full SQL queries does not yield optimal performance. However, it is remarkable that \system is able to generate near-optimal configurations for fairly small token budgets. Only extremely low settings for the token limit (196 tokens) seem to degrade performance significantly. Compared to sending full SQL queries to the LLM, \system's compression method achieves better performance even with a token reduction of more than factor ten. These results demonstrate that \system is able to prioritize effectively which information about the input workload to convey to the LLM.

\begin{figure}
    \centering
    \includegraphics[scale=0.3]{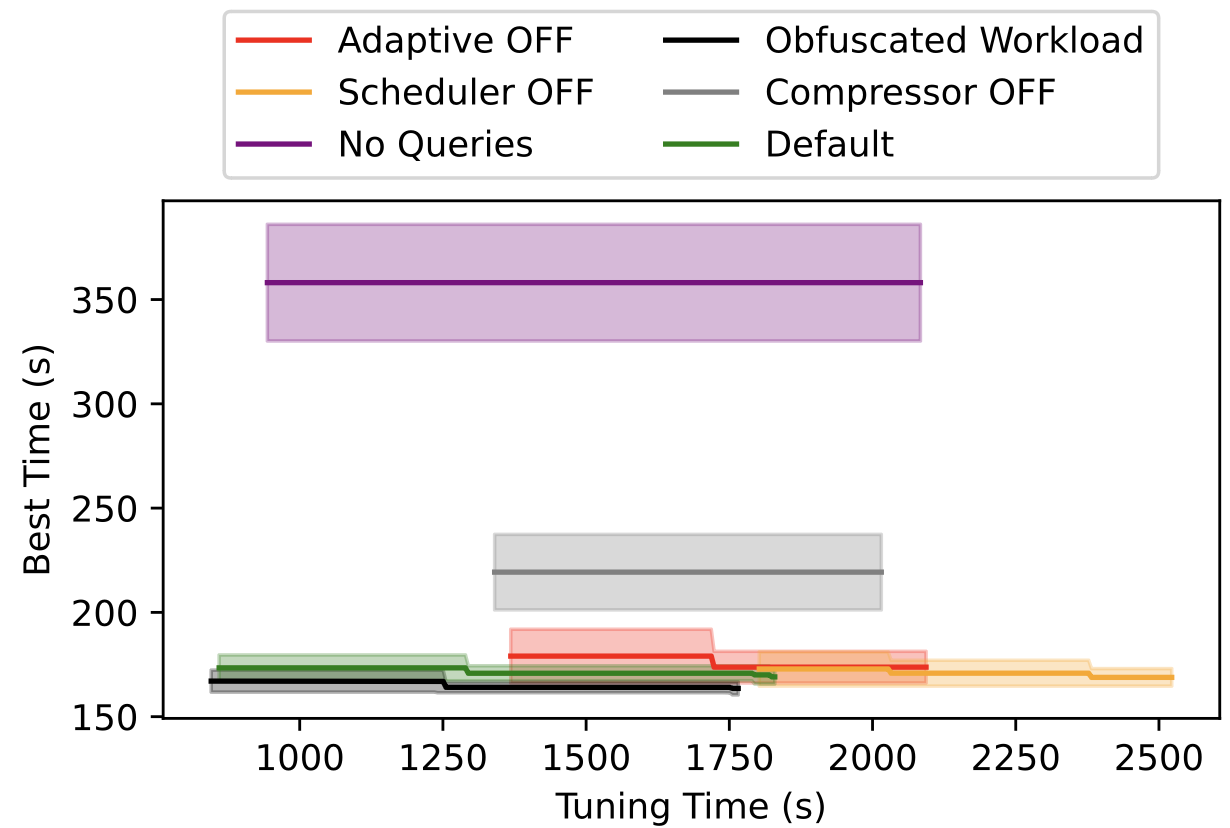}
    \caption{Ablation - JOB, Postgres, No Indexes}
    \label{fig:ablation-postgres}
\end{figure}

\begin{figure}
    \centering
    \includegraphics[width=0.85\linewidth]{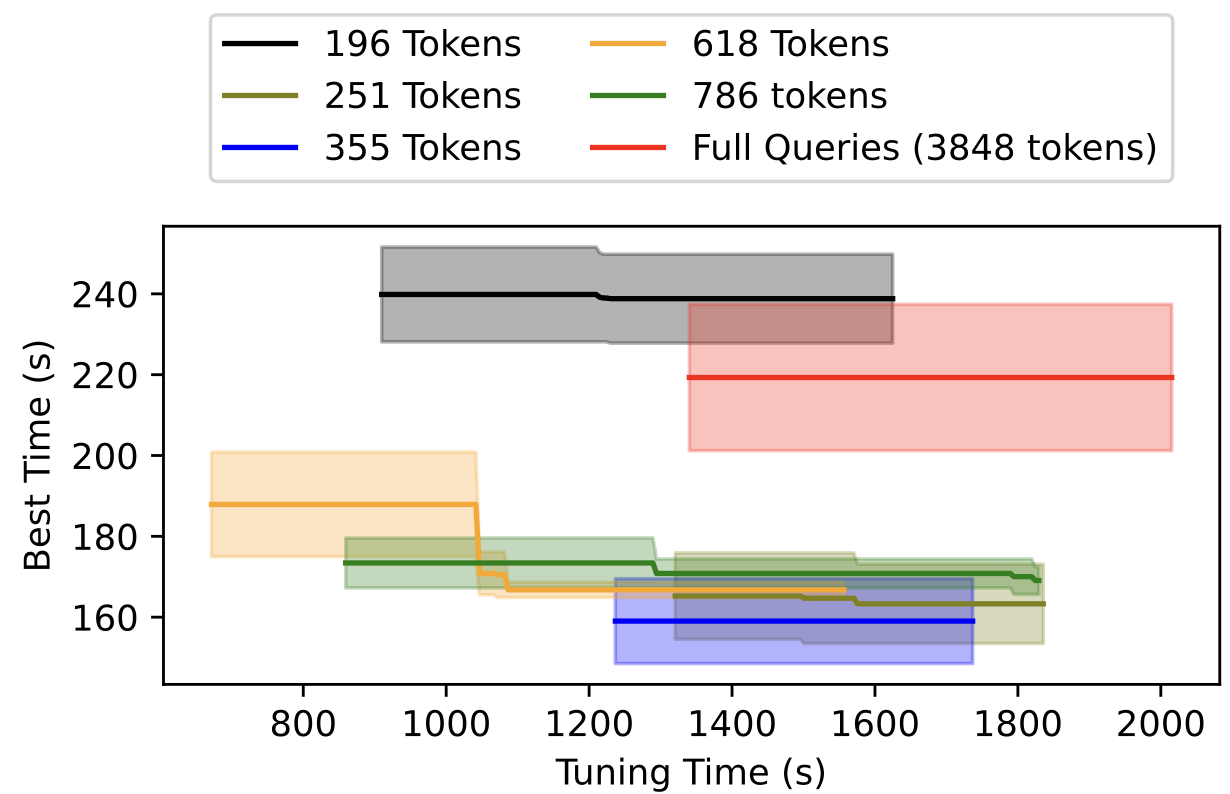}
    \caption{Ablation - Compressor Budget}
    \label{fig:ablation-compressor}
\end{figure}

\begin{figure}
    \centering
    \includegraphics[width=0.85\linewidth]{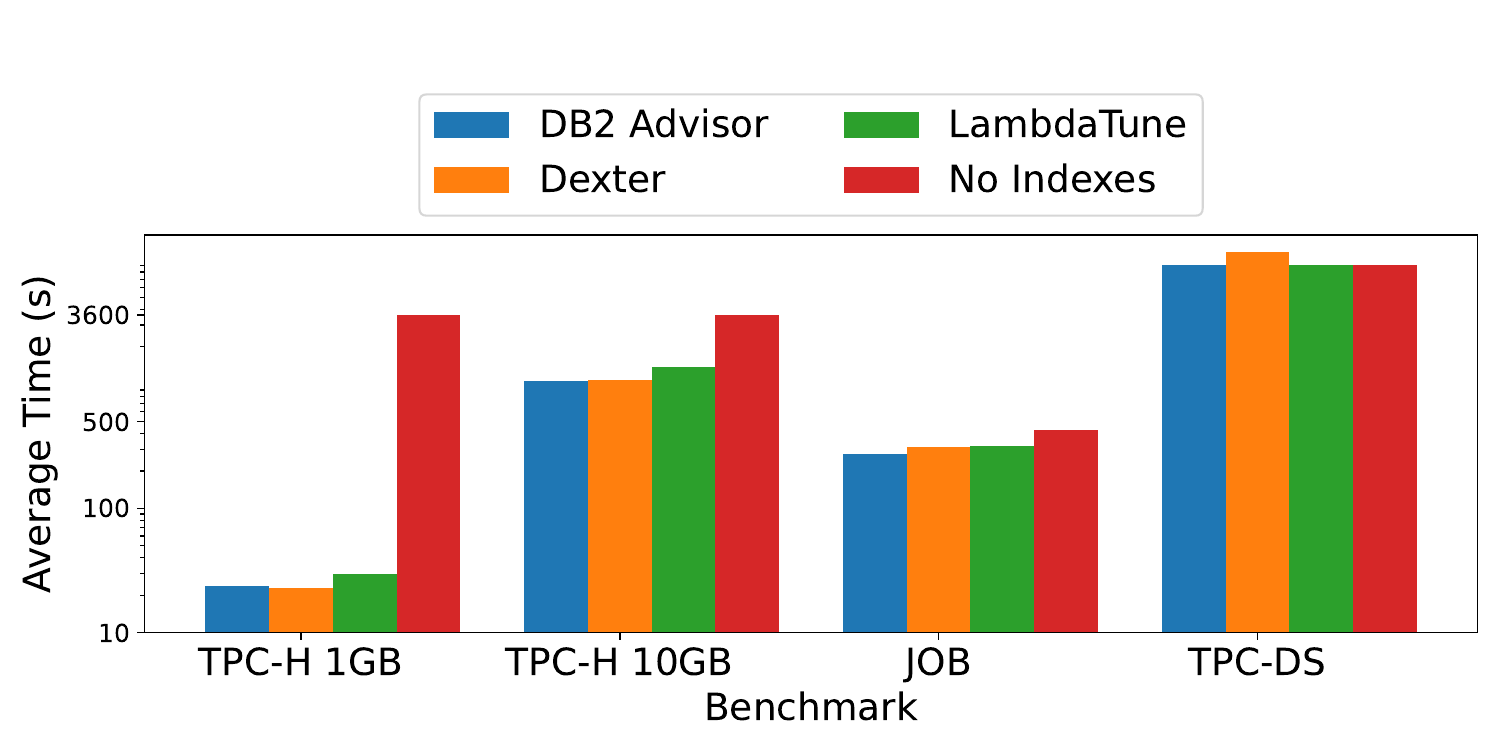}
    \caption{Comparing Index Recommendation Tools}
    \label{fig:exp:index-recommenders}
\end{figure}

\subsubsection{Index Recommendations}
Finally, Figure~\ref{fig:exp:index-recommenders} removes \system's ability to change parameter settings, focusing on index recommendation alone. The figure compares the performance of \system to different index recommendation tools (Dexter and the DB2 Index Advisor), as well as to a configuration without indexes (using the default values for all system parameters). The y-axis reports execution time and is logarithmic. Using the indexes recommended by \system reduces run time significantly, compared to the default settings. In most cases, with the exception of TPC-DS, \system does not achieve the performance of specialized index recommendation tools. This is expected, as \system has a broader scope.

%% file: related.tex
\section{Related Work}
\label{sec:related}
\label{related}
\system leverages recent advances in language models and follows the tradition of learned tuning frameworks. Our related work falls into two categories:

\textbf{DBMS Tuning.} Research efforts have focused on automating various aspects of database systems, such as query optimization~\cite{Selinger1979, zhu2023lero, marcus2022bao, anneser2023autosteer, yang2022balsa, giannakouris2022building, Yang2023, zhang2023simple}, configuration tuning~\cite{Aken, pavlo2017self, wang2021udo, kanellis2020too, kanellis2022llamatune}, and index recommendations~\cite{ding2019ai, Kanellis2022, Valentin2000a}. Traditional tuning methods~\cite{Selinger1979, Valentin2000a, Kanellis2022} are based on cost models, leveraging data statistics to estimate processing overheads with different tuning options. This may lead to sub-optimal tuning choices in case of cost estimation errors. Also, it fails for scenarios such as system parameter tuning where cost models are generally unavailable. Systems such as OtterTune~\cite{Aken}, LlamaTune~\cite{Kanellis2022}, or the ParamTree method~\cite{Yang2023} use machine learning for database tuning, guided by performance measurements via trial runs. LlamaTune exploits techniques for dimensionality reduction while ParamTree exploits existing cost models to reduce the number of required training samples. Instead, \system leverages information contained in text documents as a means to reduce the number of required trials in database tuning. \system's evaluation component uses timeouts to reduce evaluation overheads. This connects to prior methods used to reduce evaluation overheads, e.g., by selecting subsets of queries~\cite{wang2021udo, Siddiqui2022a} or by substituting calls to a classical cost model with invocations of cheaper models~\cite{Kanellis2022}. However, \system reduces evaluation overheads only via timeouts, set in order to guarantee that the system identifies the optimal configuration on the entire workload, out of all configuration generated by the language model.

\textbf{Large Language Models.} Recent work used LLMs for database system performance debugging~\cite{singh2024panda} and tuning~\cite{trummer2022db, lao2023gptuner, thakkar2024can}, with DB-BERT and GPTuner being the most relevant to our work. Both systems use LLMs to extract single hints but must navigate a vast, combinatorial space of tuning hints, requiring numerous trials to find an optimal configuration. \system becomes more efficient by using the LLM to generate entire configurations. This means we avoid combinatorial search as the LLM already combines hints about different tuning knobs. Also, this approach introduces new challenges, compared to prior work, e.g., in reducing the amount of tokens that are sent to the LLM to describe the tuning context.

%% file: conclusions.tex
\section{Conclusion}
\label{sec:conclusions}

This paper tested the following hypothesis: \emph{Sampling the output distribution of state-of-the-art LLMs, given a database tuning problem as input, yields at least some efficient configurations.} Our experiments support this hypothesis, even for newly generated workloads that do not appear in the LLM training data. Also, our experiments show that this approach is only effective if LLMs are embedded into a framework that \emph{bounds query evaluation overheads to deal with bad configurations} (via the configuration selection component discussed in Section~\ref{sec:config_select}), \emph{minimizes system configuration overheads} (via the smart evaluation component discussed in Section~\ref{sec:config_eval}), and \emph{limits monetary fees due to LLM invocations} (via the prompt compression approach discussed in Section~\ref{sec:prompt-gen}). The resulting system, \system, outperforms a variety of baselines, including several previously proposed approaches that exploit LLMs for database tuning as well.